\documentclass[aps,pra,reprint,superscriptaddress]{revtex4-2}

\usepackage{graphicx}

\usepackage{amsmath}
\usepackage{amsthm}
\usepackage{amsfonts}
\usepackage{amssymb}
\usepackage{mathtools}
\usepackage{mathrsfs}
\usepackage{latexsym}
\usepackage{bm}
\usepackage{physics}
\usepackage{tikz}
\usepackage{comment}

\usetikzlibrary{arrows.meta}

\theoremstyle{plain}
\newtheorem{thm}{Theorem}

\newtheorem{lem}[thm]{Lemma}
\newtheorem{prop}[thm]{Proposition}

\theoremstyle{definition}
\newtheorem{dfn}{Definition}
\newtheorem{ex}{Example}

\theoremstyle{remark}
\newtheorem*{rmk}{Remark}

\begin{document}


\title{Intersubjectivity as a principle determining physical observables and non-classicality}

\author{Shun Umekawa}
\email{umeshun2003@g.ecc.u-tokyo.ac.jp}
\affiliation{Department of Physics, The University of Tokyo, 5-1-5 Kashiwanoha, Kashiwa, Chiba 277-8574, Japan}
 
\author{Koki Ono}
\email{ono-koki667@g.ecc.u-tokyo.ac.jp}
\affiliation{Department of Basic Science, The University of Tokyo, 3-8-1 Komaba, Meguro-ku, Tokyo 153-8902, Japan}
 
\author{Hayato Arai}
\email{h.arai6626@gmail.com\\S. U. and K. O. contributed equally to this work.}
\affiliation{Department of Basic Science, The University of Tokyo, 3-8-1 Komaba, Meguro-ku, Tokyo 153-8902, Japan}
\affiliation{(current) Communication Science Laboratories, NTT, Inc.,
3-1 Morinosato Wakamiya, Atsugi, Kanagawa, Japan}


\begin{abstract}
We identify an operational principle that singles out Projection-Valued Measures (PVMs) among general Positive Operator-Valued Measures (POVMs), bridging the modern quantum measurement theory and the traditional formulation based on projective measurements of physical observables.
We reformulate Ozawa's intersubjectivity condition, which requires inter-observer agreement of the measurement outcomes, in a quantitative manner within the framework of generalized probabilistic theories.
We prove that (i) a POVM is a PVM if and only if its every coarse-graining is intersubjective, and (ii) a system is classical if and only if intersubjectivity is preserved under any coarse-graining, establishing a complete characterization of the physical observables and the classical theory.
Furthermore, measurements with intersubjectivity are sufficiently rich for the informational tasks of state tomography and state discrimination, testifying to its operational significance in quantum and beyond information processing.
\end{abstract}

\maketitle


\paragraph*{\textbf{Introduction.---}}
The discovery of quantum theory revealed the inherently probabilistic nature of the physical world, opening the way for measurement theory and an operational understanding of physical theories.
Traditionally, quantum theory has been formulated in terms of physical observables represented by self-adjoint operators.
In this framework, the measurement of an observable is described by a projective measurement obeying the Born rule.
In contrast, modern quantum measurement theory adopts an operational perspective centered on Positive Operator-Valued Measures (POVMs), which encompass the most general class of physically allowed measurements.
While many areas of physics still rely on the traditional approach, the operational framework has grown increasingly important, particularly with the rapid development of quantum information science.

From the modern perspective, the measurement of a physical observable is regarded as a special case of POVMs, namely, Projection-Valued Measures (PVMs).
However, PVMs are defined in an algebraic way, and their operational characterization remains unclear.
The aim of this Letter is to provide an operational characterization of PVMs among POVMs.
In order to verify a genuinely operational interpretation of the results, we employ the framework of Generalized Probabilistic Theories (GPTs) \cite{Ludwig1964,Davies1970,Gudder1973,Ozawa1980,Foulis1994,Janotta2014,Lami2018,Takakura2022,Plavala2023}, which retain only the operational requirements regarding states and measurements from the axioms of quantum theory while discarding the rest.
Within GPTs, the generalization of projective measurements, and therefore also that of ``physical observables," is not yet clear, and this question remains an active topic of research \cite{Gudder1996,Gudder1998,Gudder1999,Jencova2019,Mielnik1969,Araki1980,Kleinmann2014,Chiribella2014}.
Numerous studies have also attempted to quantify the ``PVM-likeness" of a POVM \cite{Carmeli2007,Busch2009,Liu2021,Mitra2022,Buscemi2024}; however, these approaches often lack a clear operational interpretation or fail to provide a complete characterization of PVMs.

In this Letter, we provide a fully operational characterization of PVMs along with their quantification by generalizing the concept of \textit{intersubjectivity} \cite{Ozawa2025,Khrennikov2024,Candeloro2026} introduced by Ozawa to GPTs.
A quantum measurement is called intersubjective if two observers performing the measurement simultaneously are guaranteed to obtain the same outcome.
Ozawa argued that such inter-observer agreement is essential for a measurement to ascertain an observer-independent value, rather than merely generating outcomes randomly (Fig.~\ref{fig:intersubjectivity}).
While Ozawa proved that every PVM is intersubjective in quantum theory \cite{Ozawa2025}, the converse implication remains an open question.

\begin{figure}
\begin{minipage}{\columnwidth}
\centering
\begin{tikzpicture}[scale=0.4,thick]
\draw(-0.5,-1)--(0.5,-1);
\draw(-0.5,-0.6)--(0.5,-0.6);
\draw(-0.5,-0.2)--(0.5,-0.2);
\draw(-0.5,0.2)--(0.5,0.2);
\draw(-0.5,0.6)--(0.5,0.6);
\draw(-0.5,1)--(0.5,1);
\fill(0,-0.225)circle[radius=0.2];
\draw[very thick](-0.5,0)circle[radius=2];
\draw[ultra thick](-1.7,-1.6)--(-2.6,-2.8);
\draw[very thick](0.5,0)circle[radius=2];
\draw[ultra thick](1.7,-1.6)--(2.6,-2.8);
\fill[gray](-4,-1.4)circle[radius=0.8];
\fill[gray](-2.75,-3.5)arc[radius=1.25,start angle=0,end angle=180]--cycle;
\fill[gray](4,-1.4)circle[radius=0.8];
\fill[gray](5.25,-3.5)arc[radius=1.25,start angle=0,end angle=180]--cycle;
\draw(-5.7,0.5)circle[radius=1];
\draw(-5.5,-0.9)circle[radius=0.25];
\draw(-5.2,-1.4)circle[radius=0.15];
\draw(-5.7,0.5)node{\Large 3};
\draw(5.7,0.5)circle[radius=1];
\draw(5.5,-0.9)circle[radius=0.25];
\draw(5.2,-1.4)circle[radius=0.15];
\draw(5.7,0.5)node{\Large 3};
\end{tikzpicture}
\\
(a) Intersubjective measurement
\\
\vspace{1em}
\begin{tikzpicture}[scale=0.4,thick]
\draw(-6,0)--(-7,0.6)--(-6.2,1.3)--(-5.2,0.7)--cycle;
\draw(-6,0)--(-6,-0.9);
\draw(-7,0.6)--(-7,-0.3)--(-6,-0.9)--(-5.2,-0.2)--(-5.2,0.7);
\fill(-6.1,0.65)circle[x radius=0.21,y radius=0.18];
\draw(6.2,0)--(5.2,0.6)--(6,1.3)--(7,0.7)--cycle;
\draw(6.2,0)--(6.2,-0.9);
\draw(5.2,0.6)--(5.2,-0.3)--(6.2,-0.9)--(7,-0.2)--(7,0.7);
\fill(6.15,0.325)circle[x radius=0.175,y radius=0.15];
\fill(6.05,0.975)circle[x radius=0.175,y radius=0.15];
\fill(6.55,0.675)circle[x radius=0.175,y radius=0.15];
\fill(5.65,0.625)circle[x radius=0.175,y radius=0.15];
\fill[gray](-3,1.6)circle[radius=0.8];
\fill[gray](-1.75,-0.5)arc[radius=1.25,start angle=0,end angle=180]--cycle;
\fill[gray](3,1.6)circle[radius=0.8];
\fill[gray](4.25,-0.5)arc[radius=1.25,start angle=0,end angle=180]--cycle;
\draw(-4.7,3.5)circle[radius=1];
\draw(-4.5,2.1)circle[radius=0.25];
\draw(-4.2,1.6)circle[radius=0.15];
\draw(-4.7,3.5)node{\Large 1};
\draw(4.7,3.5)circle[radius=1];
\draw(4.5,2.1)circle[radius=0.25];
\draw(4.2,1.6)circle[radius=0.15];
\draw(4.7,3.5)node{\Large 4};
\end{tikzpicture}
\\
(b) Randomly generating measurement
\end{minipage}

\caption{
Conceptual diagram of intersubjectivity.
Two observers simultaneously perform the same measurement.
(a) An intersubjective measurement guarantees identical outcomes for both observers, whereas (b) a generic measurement may produce random outcomes independently.
}
\label{fig:intersubjectivity}
\end{figure}
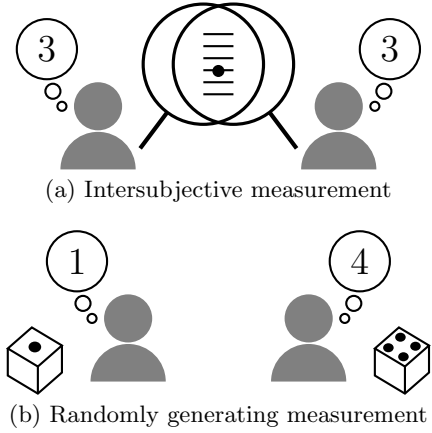

We first show that intersubjectivity alone is insufficient to characterize PVMs.
By comparing the intersubjectivity criterion with Gudder's notion of \textit{sharp effects} in GPTs \cite{Gudder1996,Gudder1998,Gudder1999} and specializing it to quantum theory, we derive a necessary and sufficient condition for a POVM to be intersubjective.
We then construct explicit examples of intersubjective POVMs that are not PVMs.
This resolves the converse direction left open in Ozawa’s analysis \cite{Ozawa2025} and reveals a genuine gap between intersubjectivity and projectivity in quantum theory.

To fill the gap between intersubjective measurements and PVMs, we introduce the notion of a \textit{completely intersubjective} measurement, namely, a measurement for which all of its coarse-grainings are intersubjective.
We show that complete intersubjectivity provides a complete characterization of PVMs in quantum theory.
Accordingly, complete intersubjectivity emerges as a natural candidate for the definition of ``physical observables" in GPTs.
Furthermore, complete intersubjectivity yields a quantification of ``PVM-likeness" that carries a direct operational interpretation.


The gap between intersubjectivity and complete intersubjectivity indicates a paradoxical situation in which the objectivity of measurement outcomes can be lost when the resolution of the measurement is reduced.
Remarkably, this counterintuitive phenomenon is common to all non-classical theories.
We prove it rigorously by constructing measurements that are intersubjective but not completely intersubjective in any non-classical theory.
Along the lines of research on GPTs that has presented characterizations of classical \cite{Barnum2007,Kuramochi2020-1,Plavala2016,Kuramochi2020-2,Aubrun2021,Aubrun2022,Shahandeh2021} and quantum \cite{Araki1980,Hardy2001,Masanes2011,Chiribella2011,Barnum2023,Arai2024} theories, this result establishes a novel operational characterization of classical theory.

Finally, to further substantiate the operational significance of (complete) intersubjectivity, we show that (completely) intersubjective measurements are sufficiently rich in any GPT in terms of standard information-processing tasks, namely, state tomography \cite{Vogel1989,Christandl2012,Banaszek2013} and state discrimination tasks \cite{Helstrom1969,Bae2015,Nuida2010,Arai2024}.
This demonstrates that (complete) intersubjectivity not only bridges the traditional observable-based formulation of quantum theory with the modern operational framework but also carries practical relevance for quantum information processing.

\paragraph*{\textbf{Generalized probabilistic theories.---}}
We briefly introduce generalized probabilistic theories (GPTs); for further details, see, \textit{e.g.}, \cite{Janotta2014,Lami2018,Takakura2022,Plavala2023}.
GPTs provide a framework based solely on operational axioms concerning states and measurements.
This framework encompasses classical and quantum theories, as well as all physically conceivable theories beyond them \cite{Popescu1994,Barrett2007,Janotta2011,Jordan1934,Sonoda2025,Yoshida2020,Plavala2022}.

In a GPT, a physical system is specified by its states and measurements.
A state is represented by an element of a compact convex set \(S\), called the state space, where the convex structure encodes probabilistic mixtures.
A extreme point of \(S\) is called a pure state, and we denote the set of all pure states by \(\mathrm{ext}(S)\).

An affine map \(a:S\to[0,1]\) is called an effect, and the set of all effects is denoted by \(\mathcal E_S\).
The state space is assumed to be separated by effects: for each \(s_1\ne s_2\in S\), there exists an effect \(a\in\mathcal E_S\) such that \(a(s_1)\ne a(s_2)\).
The partial order on \(\mathcal E_S\) is defined by \(a\le b:\Leftrightarrow(a(s)\le b(s)\text{ for all }s\in S)\). The constant maps \(0_S\) and \(1_S\), which return \(0\) and \(1\) for all states, are the least and greatest elements of \(\mathcal E_S\), respectively.

A finite-outcome measurement is given by a family \(A=(a_x)_{x\in X}\) of effects satisfying \(\sum_{x\in X}a_x=1_S\).
When a measurement \(A=(a_x)_{x\in X}\) is performed on a state \(s\in S\), the probability of obtaining the outcome \(x\in X\) is \(a_x(s)\).
We denote by \(\mathcal M_S(X)\) the set of all measurements with outcome set \(X\).
We extend several of the results for continuous-outcome measurements; see the Supplemental Material \footnote{
See Supplemental Material for further details.
}.
\newcounter{fn:SM}
\setcounter{fn:SM}{\value{footnote}}

As standard examples, the state space of a classical system is given by a simplex, while that of a quantum system is given by the space of density operators on a Hilbert space.
In quantum theory, effects correspond to positive operators bounded above by the identity operator, and measurements correspond to POVMs.

A measurement \(C=(c_{x,y})_{(x,y)\in X\times Y}\) is called a \textit{joint measurement} of a pair of measurements \(A=(a_x)_{x\in X}\) and \(B=(b_y)_{y\in Y}\) if its marginals reproduce \(A\) and \(B\), that is, \(\sum_{y\in Y}c_{x,y}=a_x\) for all \(x\in X\) and \(\sum_{x\in X}c_{x,y}=b_y\) for all \(y\in Y\).
We denote the set of all joint measurements of \(A\) and \(B\) by \(\mathrm{JM}(A,B)\).
A pair of measurements \((A,B)\) is said to be compatible \cite{Heinosaari2016} if \(\mathrm{JM}(A,B)\ne\emptyset\).

\paragraph*{\textbf{Intersubjectivity and sharpness.---}}
The concept of intersubjectivity plays a central role in this Letter.
Ozawa \cite{Ozawa2025} introduced intersubjectivity as a property of quantum observables requiring that all observers obtain the same outcome when they simultaneously measure the same observable (Fig.~\ref{fig:intersubjectivity}).
The term ``intersubjective" indicates that the measurement outcome is shared among observers, even though it is not assumed to represent a pre-existing ``objective" value independent of the measurement.
We extend this notion to GPTs and introduce a quantitative measure of intersubjectivity as follows.

\begin{dfn}
A measurement \(A=(a_x)_{x\in X}\) is called \(\alpha\)-\textit{intersubjective} if
\begin{equation}
\forall B=(b_{x,x'})_{x,x'\in X}\in\mathrm{JM}(A,A),\quad\sum_{x\in X}b_{x,x}\ge\alpha1_S.
\end{equation}
A measurement is said to be \textit{intersubjective} when it is \(1\)-intersubjective.
\end{dfn}

This definition admits a direct operational interpretation: a measurement \(A\) is \(\alpha\)-intersubjective if, whenever two observers perform it simultaneously, they agree on the outcome with a probability of at least \(\alpha\).

To analyze intersubjectivity, we also introduce a closely related concept: \textit{sharpness}.
Sharp effects were proposed by Gudder \cite{Gudder1996,Gudder1998,Gudder1999} as a generalization of projection operators in quantum theory, motivated by quantum logic \cite{Birkhoff1936,Foulis1994}.
Extending this idea, we define the sharpness of a measurement as follows.

\begin{dfn}
A measurement \(A=(a_x)_{x\in X}\) is called \(\alpha\)-\textit{sharp} if, for any \(x\ne x'\in X\),
\begin{equation}
\forall b\in\mathcal E_S,\quad(b\le a_x\text{ and }b\le a_{x'})\implies b\le(1-\alpha)1_S.
\end{equation}
A measurement is said to be \textit{sharp} when it is \(1\)-sharp.
\end{dfn}

We call an effect \(a\) sharp if the two-outcome measurement \(A:=(a,1_S-a)\) is sharp, thereby recovering Gudder's original definition \cite{Gudder1996}.
We say that a measurement \(A=(a_x)_{x\in X}\) is \textit{elementwise sharp} if \(a_x\) is sharp for all \(x\in X\).

We now show that the conditions of intersubjectivity and sharpness, although conceptually distinct at first glance, are, in fact, equivalent in any physical theory:

\begin{prop}
\label{intersubjective iff sharp}
A measurement is intersubjective if and only if it is sharp.
\end{prop}

Moreover, this equivalence admits a quantitative refinement, which is presented in the End Matter.

\paragraph*{\textbf{Characterization of PVMs.---}}
Although Ozawa \cite{Ozawa2025} proved that all PVMs are intersubjective and that not all POVMs are intersubjective, it remained an open question whether there exists an intersubjective POVM other than PVMs.
By applying Proposition~\ref{intersubjective iff sharp} to quantum theory, we find the following theorem.

\begin{thm}
\label{equivalent condition to be intersubjective}
A POVM \(A=(a_x)_{x\in X}\) is intersubjective if and only if, for any \(x\ne x'\in X\),
\begin{equation}
\mathrm{supp}(a_x)\cap\mathrm{supp}(a_{x'})=\{0\}.
\end{equation}
\end{thm}

(This holds even in the infinite-dimensional case; see the Supplemental Material \footnotemark[\value{fn:SM}].)
Applying Theorem \ref{equivalent condition to be intersubjective}, we can find examples of intersubjective POVMs that are not PVMs.
For instance, in a qubit system, the POVM \(A=\bigl(\frac12\ketbra0{0},\frac12\ketbra1{1},\frac12\ketbra+{+},\frac12\ketbra-{-}\bigr)\) satisfies the condition of Theorem \ref{equivalent condition to be intersubjective}.
This shows that intersubjectivity alone does not completely characterize PVMs.

To fill this gap, we introduce the notion of \textit{completely intersubjective} measurements by focusing on the coarse-graining of measurements.
A measurement \(B=(b_y)_{y\in Y}\) is called a \textit{coarse-graining} of a measurement \(A=(a_x)_{x\in X}\) if there exists a map \(\pi:X\to Y\) such that \(\sum_{x\in\pi^{-1}(\{y\})}a_x=b_y\) for all \(y\in Y\).

\begin{dfn}
A measurement is called \textit{completely (\(\alpha\)-)intersubjective} if every coarse-graining of the measurement is (\(\alpha\)-)intersubjective.
\end{dfn}

In quantum theory, an effect is sharp if and only if it is represented by a projection operator.
Using this fact, we obtain the following characterization.

\begin{thm}
\label{completely intersubjective iff PVM}
A POVM is a PVM if and only if it is completely intersubjective.
\end{thm}

This theorem establishes that complete intersubjectivity provides an operational characterization of PVMs, which represent physical observables in quantum theory.
Complete intersubjectivity requires that measurement outcomes be shared by all observers, independently of the measurement resolution; such observer-independent objectivity is a basic property expected of any ``physical observable."
Consequently, in GPTs where PVMs are not defined, completely intersubjective measurements naturally arise as a reasonable candidate for the proper definition of ``physical observables."

Furthermore, the quantification of complete intersubjectivity satisfies properties desirable for a measure of ``observable-likeness" of measurements: it admits a direct operational interpretation and completely characterizes PVMs in quantum theory.

\paragraph*{\textbf{Quantitative properties of intersubjectivity.---}}
Next, we discuss the basic properties of the degree of intersubjectivity.
First, we describe its behavior under joint measurements and probabilistic mixtures:
\begin{itemize}
\item [(i)] 
If measurements \(A\) and \(B\) are \(\alpha\)- and \(\beta\)-intersubjective, respectively, then any joint measurement \(C\in\mathrm{JM}(A,B)\) is \((\alpha+\beta-1)\)-intersubjective.
\item [(ii)]
Consider a probabilistic mixture \(C=\lambda A+(1-\lambda)B\) with \(\lambda\in(0,1]\), where \(A=(a_x)_{x\in X}\) and \(B=(b_x)_{x\in X}\).
If \(A\) is not (completely) \(\alpha\)-intersubjective, then \(C\) is not (completely) \((1-\lambda(1-\alpha))\)-intersubjective.
\end{itemize}

Here, the probabilistic mixture of two measurements \(A=(a_x)_{x\in X}\) and \(B=(b_x)_{x\in X}\) is defined by \(\lambda A+(1-\lambda)B:=(\lambda a_x+(1-\lambda)b_x)_{x\in X}\) for \(\lambda\in[0,1]\).
The proof is given in the Supplemental Material \footnotemark[\value{fn:SM}].

As a special case of statement (i), a joint measurement of intersubjective measurements must be intersubjective.
However, a joint measurement of two completely intersubjective measurements is not always completely intersubjective in general, even though this holds in classical and quantum theories \cite{Davies1976,Heinosaari2008,Haapasalo2014}.
An explicit counterexample is presented in Example~\ref{ES and Ext but not CIS} in the End Matter.

Statement (ii) implies that adding noise does not improve the intersubjectivity of a measurement.
This suggests that the quantification of (complete) intersubjectivity can also be interpreted as quantitative indicators of ``noiselessness" of measurements.
Importantly, this interpretation aligns with the historical perspective: PVMs have traditionally been regarded as perfectly accurate, while POVMs are generalized measurements that may incorporate noise.
Accordingly, the quantification of ``PVM-likeness" has also been discussed in the context of measuring noise \cite{Carmeli2007,Busch2009,Liu2021,Mitra2022,Buscemi2024}.


These observations demonstrate the importance of the degree of (complete) intersubjectivity.
However, evaluating this quantity for a given measurement is generally nontrivial.
We now present results for several simple cases.
A coin-tossing measurement \(A=(\lambda_x1_S)_{x\in X}\) is (completely) \(\alpha\)-intersubjective if and only if
\begin{equation}
\alpha\le\max\Bigl\{2\max_{x\in X}\lambda_x-1,0\Bigr\}.
\end{equation}
In a classical system \(S\), a measurement \(A=(a_x)_{x\in X}\) is (completely) \(\alpha\)-intersubjective if and only if
\begin{equation}
\alpha\le\inf_{s\in\mathrm{ext}(S)}\max\Bigl\{2\max_{x\in X}a_x(s)-1,0\Bigr\}.
\end{equation}
In a qubit system, an unbiased two-outcome measurement \cite{Busch2009} \(A=\bigl(\frac12(I+\bm\lambda\cdot\bm\sigma),\frac12(I-\bm\lambda\cdot\bm\sigma)\bigr)\), where \(\bm\sigma:=(\sigma_1,\sigma_2,\sigma_3)\) denotes the Pauli matrices, is (completely) \(\alpha\)-intersubjective if and only if
\begin{equation}
\alpha\le\abs{\bm\lambda}^2.
\end{equation}
The proofs of these results are provided in the Supplemental Material \footnotemark[\value{fn:SM}].

\paragraph*{\textbf{Relationships among conditions.---}}
We now explore the relationships among similar conditions for measurements: intersubjectivity, complete intersubjectivity, elementwise sharpness, and extremality in the space of measurements.
In classical theory, all of these conditions are equivalent, whereas in GPTs, they can diverge (Fig.~\ref{fig:relationships among the conditions}).

\begin{figure*}
\centering
\begin{minipage}{5.5cm}
\centering
\begin{tikzpicture}
\draw[very thick] (0,0)rectangle(5,4);
\fill[yellow,opacity=0.2,rounded corners] (0.5,3.5)rectangle(4.5,0.4);
\fill[magenta,opacity=0.08,rounded corners] (2,2.1)rectangle(4.1,0.8);
\fill[white,opacity=0.2,rounded corners] (1.5,2.5)rectangle(2.6,1.6);
\fill[cyan,opacity=0.2,rounded corners] (1.5,2.5)rectangle(2.6,1.6);
\draw[rounded corners,semithick] (2.25,3.5)--(0.5,3.5)--(0.5,0.4)--(4.5,0.4)--(4.5,3.5)--(2.75,3.5);
\draw (2.5,3.47) node{\textsf{IS}};
\draw[rounded corners,semithick] (3.2,2.1)--(2,2.1)--(2,0.8)--(4.1,0.8)--(4.1,2.1)--(3.9,2.1);
\draw (3.55,2.07) node{\textsf{Ext}};
\draw[rounded corners,semithick] (1.7,3)--(1,3)--(1,1.2)--(3,1.2)--(3,3)--(2.3,3);
\draw (2,2.97) node{\textsf{ES}};
\draw[rounded corners,semithick] (1.7,2.5)--(1.5,2.5)--(1.5,1.6)--(2.6,1.6)--(2.6,2.5)--(2.4,2.5);
\draw (2.05,2.47) node{\textsf{CIS}};
\end{tikzpicture}\\
(a) GPTs
\end{minipage}
\begin{minipage}{5.5cm}
\centering
\begin{tikzpicture}
\draw[very thick] (0,0)rectangle(5,4);
\fill[yellow,opacity=0.2,rounded corners] (0.5,3.5)rectangle(4.5,0.4);
\fill[magenta,opacity=0.08,rounded corners] (1.0,2.9)rectangle(4.0,0.9);
\fill[white,opacity=0.2,rounded corners] (1.5,2.3)rectangle(3.5,1.4);
\fill[cyan,opacity=0.2,rounded corners] (1.5,2.3)rectangle(3.5,1.4);
\draw[rounded corners,semithick] (2.25,3.5)--(0.5,3.5)--(0.5,0.4)--(4.5,0.4)--(4.5,3.5)--(2.75,3.5);
\draw (2.5,3.47) node{\textsf{IS}};
\draw[rounded corners,semithick] (2.15,2.9)--(1.0,2.9)--(1.0,0.9)--(4.0,0.9)--(4.0,2.9)--(2.85,2.9);
\draw (2.5,2.87) node{\textsf{Ext}};
\draw[rounded corners,semithick] (1.75,2.3)--(1.5,2.3)--(1.5,1.4)--(3.5,1.4)--(3.5,2.3)--(3.25,2.3);
\draw (2.5,2.27) node{\textsf{ES\,=\,CIS}};
\end{tikzpicture}\\
(b) Quantum theory
\end{minipage}
\begin{minipage}{5.5cm}
\centering
\begin{tikzpicture}
\draw[very thick] (0,0)rectangle(5,4);
\fill[yellow,opacity=0.2,rounded corners] (0.5,2.9)rectangle(4.5,0.9);
\fill[magenta,opacity=0.08,rounded corners] (0.5,2.9)rectangle(4.5,0.9);
\fill[white,opacity=0.2,rounded corners] (0.5,2.9)rectangle(4.5,0.9);
\fill[cyan,opacity=0.2,rounded corners] (0.5,2.9)rectangle(4.5,0.9);
\draw[rounded corners,semithick] (1.0,2.9)--(0.5,2.9)--(0.5,0.9)--(4.5,0.9)--(4.5,2.9)--(4.0,2.9);
\draw (2.5,2.87) node{\textsf{IS\,=\,Ext\,=\,ES\,=\,CIS}};
\end{tikzpicture}\\
(c) Classical theory
\end{minipage}
\caption{Inclusion relation among the set of intersubjective (IS), completely intersubjective (CIS), and  elementwise sharp (ES) measurements and the extreme points of the space of measurements (Ext) in (a) GPTs, (b) quantum theory and (c) classical theory.}
\label{fig:relationships among the conditions}
\end{figure*}
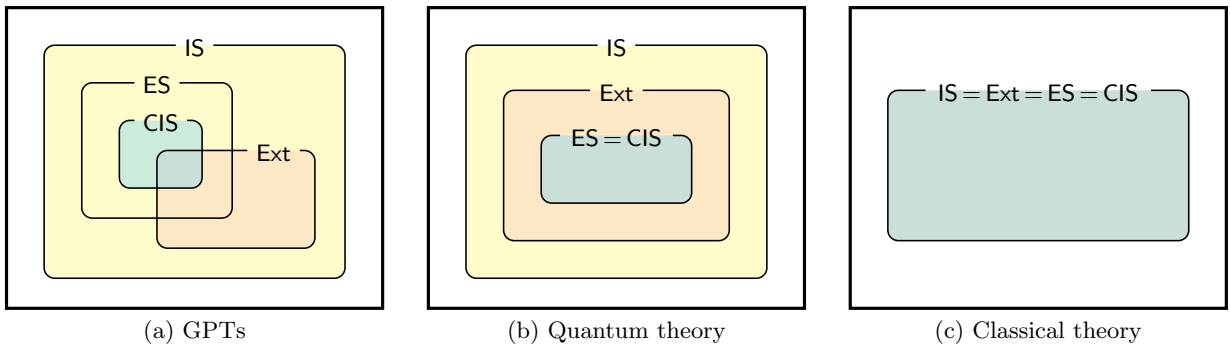

From the definitions, it is straightforward that \textit{completely intersubjective} \(\implies\) \textit{elementwise sharp} \(\implies\) \textit{intersubjective}.
Moreover, it follows from the result in \cite{Guerini2018}, which relates extremality to the uniqueness of joint measurements, that every extremal measurement \(A\in\mathrm{ext}(\mathcal M_S(X))\) is intersubjective.
This observation is significant: the Krein--Milman theorem implies that any measurement can be realized as a convex mixture of intersubjective measurements, which also indicates that intersubjective measurements can be regarded as ``noiseless" or ``accurate" measurements, while general measurements may be understood as their error-tolerant generalizations.

In quantum theory, in addition to the above general relations, the implications \textit{elementwise sharp} \(\implies\) \textit{completely intersubjective} \(\implies\) \textit{extremal} hold; in particular, complete intersubjectivity and elementwise sharpness are both equivalent to being a PVM.

Apart from the inclusion relations stated above, no other relations hold in general.
For example, in quantum theory there exist (i) extremal POVMs that are not PVMs, and (ii) intersubjective POVMs that are not extremal.
In more general systems, there exist measurements that are (iii) completely intersubjective but not extremal, (iv) elementwise sharp and extremal but not completely intersubjective, and (v) elementwise sharp but neither extremal nor completely intersubjective.
Explicit counterexamples are provided in the End Matter.

\paragraph*{\textbf{Characterization of classical theory.---}}
One of the prominent results from studies of GPTs is that certain features once regarded as specific mysteries of quantum theory, such as the existence of incompatible measurements \cite{Plavala2016, Kuramochi2020-2} and entanglement \cite{Aubrun2021,Aubrun2022}, are, in fact, common to all non-classical theories.
Building on this perspective, we show that the gap between intersubjectivity and complete intersubjectivity provides a novel operational characterization of non-classical theory.

\begin{thm}
\label{non-classical iff IS but not CIS exist}
A finite-dimensional system is classical if and only if all intersubjective measurements on it are completely intersubjective.
\end{thm}

For the proof, see the End Matter.

This result demonstrates that in any non-classical theory, the objectivity of measurement outcomes can depend on the resolution of the measurement and may be lost through coarse-graining.
This phenomenon can be illustrated by the following paradoxical thought experiment.
Suppose that there are three boxes, and a particle is located in one of them.
Two observers, Alice and Bob, attempt to find where the particle is.
Assume that the measurement of the location of the particle is intersubjective but not completely intersubjective.
Then the following counterintuitive situation arises: if they open all three boxes and simultaneously observe the location of the particle, the intersubjectivity ensures the consensus of their observation; if they open only one box and check whether the particle is inside that box, they may arrive at different conclusions because the measurement is not completely intersubjective.
Example \ref{Ext but not PVM} in the End Matter realizes this thought experiment within quantum theory.

\paragraph*{\textbf{Advantage in information processing.---}}
Beyond the operational insights presented so far, we now clarify the significance of (complete) intersubjectivity in specific information-processing tasks: state tomography \cite{Vogel1989,Christandl2012,Banaszek2013} and state discrimination \cite{Helstrom1969,Bae2015,Nuida2010,Arai2024}, which play a fundamental role in (quantum) information processing.

State tomography is the informational task of determining an unknown state from the statistics of measurement outcomes.
A set of measurements \(\mathcal M\) is called \textit{tomographically complete} if it separates the state space; that is, for each \(s_1\ne s_2\in S\), there exists a measurement \(A=(a_x)_{x\in X}\in\mathcal M\) and \(x\in X\) such that \(a_x(s_1)\ne a_x(s_2)\).
If \(\mathcal M\) is tomographically complete, state tomography can be performed using measurements solely from \(\mathcal M\).

State discrimination is the task of distinguishing a family of states through a single-shot measurement: given a finite ensemble of states \((s_x,p_x)_{x\in X}\), where state \(s_x\) is prepared with probability \(p_x\), the goal is to identify which state has been prepared.
One performs a measurement \(A=(a_x)_{x\in X}\) on the given state and guesses that the prepared state was \(s_x\) when outcome \(x\) is obtained.
The success probability of this task is given by \(\sum_{x\in X} p_xa_x(s_x)\).
A family of states \((s_x)_{x\in X}\) is said to be \textit{perfectly distinguished} by a measurement \(A=(a_x)_{x\in X}\) if \(a_x(s_x)=1\) holds for all \(x\in X\).

We observe the following:
\begin{itemize}
\item[(i)]
The set of all completely intersubjective measurements is tomographycally complete.
\item[(ii)]
For any given ensemble of states, there always exists an intersubjective measurement that is optimal, \textit{i.e.}, that maximizes the success probability of the discrimination task.
\item[(iii)]
Every completely intersubjective measurement admits a family of states that is perfectly distinguished by it.
\end{itemize}

Statement (i) follows from the fact that every extremal two-outcome measurement is sharp \cite{Gudder1999}, while (ii) follows from Bauer's maximum principle and the implication \textit{extremal \(\implies\) intersubjective}.
Statement (iii) follows from the implication \textit{completely intersubjective \(\implies\) elementwise sharp}, along with the fact that for any sharp effect \(a\ne0_S\), there exists a state \(s\) such that \(a(s)=1\) \cite{Gudder1999}.
We note, however, that unlike in the quantum case, a perfectly distinguishable family of states need not be perfectly distinguished by a completely intersubjective measurement, as demonstrated by Example~\ref{ES and Ext but not CIS} in the End Matter.

When defining a concept in GPTs, one must keep in mind that instances of it may arise only rarely in generic systems \cite{Kleinmann2014,Jencova2019}.
Despite this concern, the above results demonstrate the abundant existence of (completely) intersubjective measurements.

\paragraph*{\textbf{Conclusion.---}}
In this work, we have investigated the operational foundations of measurements corresponding to physical observables in quantum and generalized probabilistic theories.
We have extended Ozawa’s notion of intersubjectivity and introduced complete intersubjectivity, which requires that any coarse-graining of a measurement preserves intersubjectivity.
In quantum theory, this property completely characterizes projective measurements (PVMs), providing a natural operational criterion for a measurement to represent a physical observable.

We have quantified intersubjectivity and complete intersubjectivity, providing operational measures of ``observable-likeness" and ``noiselessness" of measurements.
Our results also reveal that the gap between intersubjectivity and complete intersubjectivity serves as a distinguishing feature of non-classical theories, providing a new operational criterion to characterize classical systems.

Finally, we have shown that sufficiently many (completely) intersubjective measurements exist for state tomography and state discrimination tasks in any GPT.
This highlights their conceptual and practical roles: not only do they ensure shared outcomes among observers, but they also provide resources for information processing, even beyond quantum theory.
These results demonstrate that intersubjectivity unifies foundational and operational perspectives on measurements.

Future work may explore several directions.
First, an intriguing question is whether quantum theory can be characterized by additional inclusion relations among the notions related to intersubjectivity.
Second, it is also natural to ask whether intersubjectivity admits a resource-theoretic formulation and how it relates to measurement incompatibility.
Finally, further developing an operational understanding of broad topics in traditional observable-based physics may not only deepen its conceptual foundations but also provide a basis for formulating and exploring diverse physical theories within the framework of GPTs.


\paragraph*{\textbf{Acknowledgments.---}}
The authors thank the Yukawa Institute for Theoretical Physics at Kyoto University, where this work was initiated during the YITP-W-25-11.
We thank Hiroatsu Okamitsu, Masanao Ozawa, Ryo Takakura, Kaito Watanabe, and Ludovico Lami for fruitful discussions.
S. U. was supported by Forefront Physics and Mathematics Program to Drive Transformation, a World-leading Innovative Graduate Study Program, the University of Tokyo.
K. O. was supported by Advanced Basic Science Course, a World-leading Innovative Graduate Study Program, the University of Tokyo.
H. A. was supported by JSPS KAKENHI Grant Number 25KJ0043 and 26K17031.

\paragraph*{\textbf{Author contributions.---}}
S. U. and K. O. contributed equally to this work.

\bibliography{intersubjectivity_ref}

\clearpage

\section*{End Matter}

\paragraph*{\textbf{Proof of Proposition \ref{intersubjective iff sharp}.---}}
Here, we present the proofs of Proposition \ref{intersubjective iff sharp} and its quantitative extension.

\begin{lem}
\label{quantitative relation between intersubjectivity and shapness}
\leavevmode
\begin{itemize}
\item[(i)]
If an \(n\)-outcome measurement is \(\alpha\)-sharp, then it is \(\bigl(1-\bigl(n^2-n\bigr)(1-\alpha)\bigr)\)-intersubjective.
\item[(ii)]
If a measurement is \(\alpha\)-intersubjective, then it is \(\frac{1+\alpha}2\)-sharp.
\end{itemize}
\end{lem}

\begin{proof}
\leavevmode
\begin{itemize}
\item[(i)]
Let \(A=(a_x)_{x\in X}\) be an \(n\)-outcome measurement that is \(\alpha\)-sharp, and let \(B=(b_{x,x'})_{x,x'\in X}\in\mathrm{JM}(A,A)\).
Then we have \(b_{x,x'}\le a_x\) and \(b_{x,x'}\le a_{x'}\) for all \(x,x'\in X\).
Since \(A\) is \(\alpha\)-sharp, it follows that \(b_{x,x'}\le(1-\alpha)1_S\) for \(x\ne x'\).
Therefore, 
\begin{equation*}
\begin{aligned}
\sum_{x\in X}b_{x,x}&=1_S-\sum_{\substack{x,x'\in X\\x\ne x'}}b_{x,x'}\\
&\ge\bigl(1-(n^2-n)(1-\alpha)\bigr)1_S.
\end{aligned}
\end{equation*}
Hence, \(A\) is \(\bigl(1-\bigl(n^2-n\bigr)(1-\alpha)\bigr)\)-intersubjective.
\item[(ii)]
Let \(A=(a_x)_{x\in X}\) be a measurement that is \(\alpha\)-intersubjective.
Suppose an effect \(c\) satisfies \(c\le a_{x_1}\) and \(c\le a_{x_2}\) for some \(x_1\ne x_2\in X\).
Define \(B=(b_{x,x'})_{x,x'\in X}\in\mathrm{JM}(A,A)\) by
\begin{equation*}
b_{x,x'}:=
\begin{cases}
a_{x_1}-c&\text{if }x=x'=x_1,\\
a_{x_2}-c&\text{if }x=x'=x_2,\\
a_x&\text{if }x=x'\ne x_1,x_2,\\
c&\text{if }x=x_1\text{ and }x'=x_2,\\
c&\text{if }x=x_2\text{ and }x'=x_1,\\
0_S&\text{otherwise}.
\end{cases}
\end{equation*}
Since \(A\) is \(\alpha\)-intersubjective, we have
\begin{equation*}
c=\frac12\biggl(1_S-\sum_{x\in X}b_{x,x}\biggr)\le\frac{1-\alpha}21_S.
\end{equation*}
Hence, \(A\) is \(\frac{1+\alpha}2\)-sharp.
\qedhere
\end{itemize}
\end{proof}

Proposition \ref{intersubjective iff sharp} follows from Lemma \ref{quantitative relation between intersubjectivity and shapness}.
Moreover, for two-outcome measurements (\textit{i.e.}, \(n=2\)), the quantifications of intersubjectivity and sharpness coincide: a two-outcome measurement is \(\alpha\)-intersubjective if and only if it is \(\frac{1+\alpha}2\)-sharp.

\paragraph*{\textbf{Counterexamples.---}}
Here, we present several explicit counterexamples mentioned in the main text.

\begin{ex}
\label{Ext but not PVM}
In a qubit system, the POVM \(A:=\bigl(\frac23\ketbra0{0},\frac23\ketbra{\phi_+}{\phi_+},\frac23\ketbra{\phi_-}{\phi_-}\bigr)\), where \(\ket{\phi_\pm}:=\frac12\ket0\pm\frac{\sqrt3}2\ket1\), is an extreme point of \(\mathcal M_S(3)\) but not a PVM.
\end{ex}

\begin{ex}
\label{IS but not Ext POVM}
In a qubit system, the mixture \(A:=\frac12B+\frac12C\) of two POVMs \(B:=(\ketbra0{0},\ketbra1{1},0,0)\) and \(C:=(0,0,\ketbra+{+},\ketbra-{-})\) is intersubjective but not extremal.
\end{ex}

\begin{ex}
\label{CIS but not Ext}
Consider the square model (also known as a gbit \cite{Barrett2007}), a model that serves as the local component of the Popescu-Rohrlich box \cite{Popescu1994}, with the state space
\begin{equation*}
S:=\left\{(x_1,x_2,x_3)\in\mathbb R^3\:\middle|\:
\begin{gathered}
-x_3\le x_1\le x_3,\\
-x_3\le x_2\le x_3,\\
x_3=1
\end{gathered}
\right\}.
\end{equation*}
Let \(A:=\frac12B+\frac12C\) be the mixture of two measurements \(B:=(b_+,b_-)\) and \(C:=(c_+,c_-)\), defined by
\begin{equation*}
\begin{aligned}
b_\pm(x_1,x_2,x_3)&:=\frac12(x_3\pm x_1),\\
c_\pm(x_1,x_2,x_3)&:=\frac12(x_3\pm x_2).
\end{aligned}
\end{equation*}
Then \(A\) is completely intersubjective but not extremal; see the Supplemental Material \footnotemark[\value{fn:SM}] for the proof.
\end{ex}

\begin{ex}
\label{ES and Ext but not CIS}
Consider the system with the state space
\begin{equation*}
\begin{aligned}
&S:=\\
&\left\{(x_1,x_2,x_3,x_4,x_5)\in\mathbb R^5\:\middle|\:
\begin{gathered}
0\le x_i\ (i=1,\dots,5),\\
x_5\le x_1+x_4,\\
x_5\le x_2+x_3,\\
x_1+x_2+x_3+x_4=1
\end{gathered}
\right\}.
\end{aligned}
\end{equation*}
Define the measurement \(A:=(a_1,a_2,a_3,a_4)\) by
\begin{equation*}
a_i(x_1,x_2,x_3,x_4,x_5):=x_i\quad(i=1,2,3,4).
\end{equation*}
This measurement is elementwise sharp and extremal, but not completely intersubjective.

Moreover, \(A\) is a joint measurement of two completely intersubjective measurements \(B:=(a_1+a_2,a_3+a_4)\) and \(C:=(a_1+a_3,a_2+a_4)\), but \(A\) itself is not completely intersubjective.
We also find that the family of states \((s_1,s_2,s_3,s_4)\), defined by \(s_i:=(\delta_{1,i},\delta_{2,i},\delta_{3,i},\delta_{4,i},0)\), is perfectly distinguished only by \(A\), which is not completely intersubjective.
Proofs are provided in the Supplemental Material \footnotemark[\value{fn:SM}].
\end{ex}

\begin{ex}
\label{ES but not Ext nor CIS}
In the system given by the direct sum of the state spaces given in Examples \ref{CIS but not Ext} and \ref{ES and Ext but not CIS}, the measurement obtained as the direct sum of the measurements in Examples \ref{CIS but not Ext} and \ref{ES and Ext but not CIS} is elementwise sharp but neither extremal nor completely intersubjective.
Here, the direct sum of state spaces \(S_1\) and \(S_2\) is defined (informally) by
\begin{equation*}
S_1\oplus S_2:=\{\lambda s_1+(1-\lambda)s_2\mid\lambda\in[0,1],\ s_1\in S_1,\ s_2\in S_2\},
\end{equation*}
and for measurements \(A=(a_x)_{x\in X}\in\mathcal M_{S_1}(X)\) and \(B=(b_y)_{y\in Y}\in\mathcal M_{S_2}(Y)\), their direct sum \(A\oplus B=(c_z)_{z\in X\sqcup Y}\in\mathcal M_{S_1\oplus S_2}(X\sqcup Y)\) is defined by
\begin{equation*}
c_z(\lambda s_1+(1-\lambda)s_2):=
\begin{cases}
\lambda a_z(s_1)&\text{if }z\in X,\\
(1-\lambda)b_z(s_2)&\text{if }z\in Y.
\end{cases}
\end{equation*}
\end{ex}

\begin{rmk}
We note that the gap between intersubjectivity and complete intersubjectivity arises only for measurements with three or more outcomes, whereas the gap between elementwise sharpness and complete intersubjectivity arises only for measurements with four or more outcomes.
\end{rmk}

\paragraph*{\textbf{Proof of Theorem \ref{non-classical iff IS but not CIS exist}.---}}
In a classical system, there is no gap between intersubjectivity and complete intersubjectivity, as shown by the explicit evaluation of the degree of (complete) intersubjectivity given in the main text (see the Supplemental Material \footnotemark[\value{fn:SM}] for the proof).
Conversely, in any non-classical system, we explicitly construct an intersubjective and not completely intersubjective measurement in two different ways (Lemma \ref{3-outcome IS but CIS} and Lemma \ref{n+1-outcome IS but CIS}).

For the explicit constructions, we focus on the following concept: \textit{indecomposable effects}.
An effect \(a\ne0_S\) is said to be indecomposable \cite{Kimura2010} if
\begin{equation}
\forall b\in\mathcal E_S,\quad b\le a \implies \exists \lambda\in[0,1]\text{ s.t. }b=\lambda a.
\end{equation}
Indecomposable effects correspond to the extremal rays of the positive cone of affine functionals on the state space.
A finite-dimensional system is classical if and only if all indecomposable effects of unit norm are linearly independent (see, \textit{e.g.}, §2.2 of \cite{Aubrun2021}).
Here, the norm of an effect \(a\) is given by \(\norm{a}:=\sup_{s\in S}a(s)\).

\begin{lem}
\label{3-outcome IS but CIS}
In any finite-dimensional non-classical system, there exists a three-outcome measurement that is intersubjective but not completely intersubjective.
\end{lem}
\begin{proof}
Suppose that every three-outcome intersubjective measurement on a system \(S\) is completely intersubjective.
We show that this forces \(S\) to be classical.

First, we show that for any pair of indecomposable effects \(a\) and \(b\) of unit norm, the inequality \(a+b\le1_S\) holds.
Let \(\lambda,\mu\in[0,1]\) be chosen to maximize \(\lambda+\mu\) under the constraint \(\lambda a+\mu b\le1_S\).
Without loss of generality, we may assume that \(\lambda>0\) and \(\mu>0\); otherwise, one can take \(\lambda=\mu=\frac12\) instead.
Since \(a\) and \(b\) are indecomposable, the three-outcome measurement \(A=(\lambda a,\mu b,1-\lambda a-\mu b)\) is intersubjective.
By the assumption, the two-outcome coarse-graining \(B=(\lambda a,1-\lambda a)\) is also intersubjective.
This implies that \(\lambda=1\), and similarly \(\mu=1\); hence \(a+b\le1_S\).

Next, we show that \(S\) is classical.
Let \((a_i)_{i=1}^n\) be a mutually distinct family of indecomposable effects of unit norm.
Since \(\norm{a_i}=1\), there exists \(s_i\in S\) such that \(a_i(s_i)=1\) for each \(i\).
Since \(a_i+a_j\le1_S\) for \(i\ne j\), we have \(a_i(s_j)=\delta_{i,j}\).
Therefore, if \(\sum_{i=1}^{n}\lambda_ia_i=0_S\) with \(\lambda_i\in\mathbb R\), evaluating at \(s_j\) gives \(\lambda_j=\sum_{i=1}^{n}\lambda_ia_i(s_j)=0\).
This implies that all indecomposable effects are linearly independent.
Hence, \(S\) is classical.
\end{proof}

\begin{lem}
\label{n+1-outcome IS but CIS}
In any \(n\)-dimensional non-classical system, there exists a measurement with \(n+1\) or more outcomes that is intersubjective but not completely intersubjective.
\end{lem}
\begin{proof}
We first show that a completely intersubjective measurement on an \(n\)-dimensional system \(S\) can have at most \(n\) outcomes.
Suppose a measurement \((a_x)_{x\in X}\) has states \((s_x)_{x\in X}\) such that \(a_x(s_x)=1\) for each \(x\in X\).
Since \(\sum_{x\in X} a_x=1_S\), we have \(a_x(s_{x'})=\delta_{x,x'}\).
If \(\sum_{x\in X}\lambda_x a_x=0_S\) with \(\lambda_x\in\mathbb{R}\), evaluating at \(s_{x'}\) yields \(\lambda_{x'}=\sum_{x\in X}\lambda_x a_x(s_{x'})=0\), so the effects \((a_x)_{x\in X}\) are linearly independent; hence, \(|X|\leq n\).

In addition, every completely intersubjective measurement is elementwise sharp, and each sharp effect \(a\) has a state \(s\) with \(a(s)=1\) \cite{Gudder1999}.
Combining this with the above argument, we conclude that every completely intersubjective measurement on \(S\) has at most \(n\) outcomes.

Now, it remains to construct an intersubjective measurement with \(n+1\) or more outcomes on any non-classical system \(S\).
There exist \(n+1\) distinct indecomposable effects \((a_i)_{i=1}^{n+1}\) of unit norm because \(S\) is non-classical.
Since \(1_S\) is an interior point of the positive cone of affine functionals on \(S\), there exists \(\varepsilon>0\) such that \(1_S-\varepsilon\sum_{i=1}^{n+1}a_i\in \mathcal E_S\).
By the Krein--Milman theorem, there exist indecomposable effects \((b_j)_{j=1}^{m}\) such that \(1_S-\varepsilon\sum_{i=1}^{n+1}a_i=\sum_{j=1}^{m}b_j\).
We then construct a measurement \(C=(c_k)_{k=1}^{l}\) by combining \((\varepsilon a_i)_{i=1}^{n+1}\) and \((b_j)_{j=1}^{m}\), summing together any elements that are scalar multiples of each other.
Since \(C\) is composed of indecomposable effects that are not scalar multiples of each other, it is a sharp measurement.
Since \(l\ge n+1\), this is the desired intersubjective measurement with \(n+1\) or more outcomes.
\end{proof}


\clearpage
\appendix
\widetext
\renewcommand{\thesection}{S\arabic{section}}
\renewcommand{\theequation}{S\arabic{section}.\arabic{equation}}
\setcounter{thm}{0}
\renewcommand{\thethm}{S.\arabic{thm}}
\setcounter{dfn}{0}
\renewcommand{\thedfn}{S.\arabic{dfn}}
\setcounter{figure}{0}
\renewcommand{\thefigure}{S.\arabic{figure}}

\begin{center}
\Large\bf Supplemental Material
\end{center}

\section{Characterization of intersubjective POVMs in infinite-dimensional quantum systems}

We present a characterization of intersubjective POVMs in quantum systems represented by Hilbert spaces that are not necessarily finite-dimensional.

\begin{thm}[Infinite-dimentional version of Theorem 2]
A POVM \(A=(a_x)_{x\in X}\) is intersubjective if and only if, for any \(x\ne x'\in X\),
\begin{equation}
\mathrm{supp}(a_x)\cap\mathrm{supp}(a_{x'})=\{0\}.
\end{equation}
\end{thm}

\begin{proof}
By Proposition 1, it suffices to show that for any positive operators \(a,b\ge0\), the following two conditions are equivalent:

\begin{itemize}
\item[(i)]
\(\mathrm{supp}(a)\cap\mathrm{supp}(b)=\{0\}.\)
\item[(ii)]
For any positive operator \(c\ge0\), if \(c\le a\) and \(c\le b\), then \(c=0\).
\end{itemize}

First, assume (i).
Let \(c\ge0\) satisfy \(c\le a\) and \(c\le b\).
Then \(\mathrm{supp}(c)\subseteq\mathrm{supp}(a)\) and \(\mathrm{supp}(c)\subseteq\mathrm{supp}(b)\), and hence, \(\mathrm{supp}(c)={0}\), which implies \(c=0\).

Next, assume (ii), and suppose, to the contrary, that \(\mathrm{supp}(a)\cap\mathrm{supp}(b)\ne\{0\}.\)
Let
\begin{equation*}
a=\int_0^\infty\lambda\,de_a(\lambda),\quad b=\int_0^\infty\lambda\,de_b(\lambda)
\end{equation*}
be the spectral decompositions of \(a\) and \(b\), respectively.
For \(\varepsilon>0\), define
\begin{equation*}
e_{a,\varepsilon}:=\int_\varepsilon^\infty de_a(\lambda),\quad e_{b,\varepsilon}:=\int_\varepsilon^\infty de_b(\lambda).
\end{equation*}
As \(\varepsilon\to0^+\), the projections \(e_{a,\varepsilon}\) and \(e_{b,\varepsilon}\) converge in the strong operator topology to the support projections \(p_{\mathrm{supp}(a)}\) and \(p_{\mathrm{supp}(b)}\).
Consequently,
\begin{equation*}
e_{a,\varepsilon}\wedge e_{b,\varepsilon}\xrightarrow{\varepsilon\to0^+}p_{\mathrm{supp}(a)}\wedge p_{\mathrm{supp}(b)}>0,
\end{equation*}
where \(p\wedge q\) denotes the greatest lower bound of \(p\) and \(q\) in the complete lattice of projections on the Hilbert space.
Hence, there exists \(\varepsilon>0\) such that \(e_{a,\varepsilon}\wedge e_{b,\varepsilon}>0\).
Fix such an \(\varepsilon\) and define \(c:=\varepsilon(e_{a,\varepsilon}\wedge e_{b,\varepsilon})\).
Then \(c\le\varepsilon e_{a,\varepsilon}\le a\) and \(c\le\varepsilon e_{b,\varepsilon}\le b\), while \(c>0\).
This contradicts (ii).
\end{proof}

\begin{rmk}
The above argument extends straightforwardly to more general quantum theories in which a system is described by a von Neumann algebra.
\end{rmk}

\section{Proofs of quantitative properties of intersubjectivity}

In this appendix, we prove the quantitative properties of intersubjectivity stated in the main text.
We begin with its behavior under joint measurements and probabilistic mixtures.

\begin{prop}
\leavevmode
\begin{itemize}
\item [(i)] 
Suppose that two measurements \(A\) and \(B\) are \(\alpha\)- and \(\beta\)-intersubjective, respectively.
Then any joint measurement \(C\in\mathrm{JM}(A,B)\) is \((\alpha+\beta-1)\)-intersubjective.
\item [(ii)]
Let \(C=\lambda A+(1-\lambda)B\) with \(\lambda\in(0,1]\) be a probabilistic mixture of two measurements \(A=(a_x)_{x\in X}\) and \(B=(b_x)_{x\in X}\).
If \(A\) is not (completely) \(\alpha\)-intersubjective, then \(C\) is not (completely) \((1-\lambda(1-\alpha))\)-intersubjective.
\end{itemize}
\end{prop}

\begin{proof}
\leavevmode
\begin{itemize}
\item[(i)]
Let \(D=(d_{x,y,x',y'})_{(x,y),(x',y')\in X\times Y}\in\mathrm{JM}(C,C)\).
Define two measurements \(E=(e_{x,x'})_{x,x'\in X}\) and \(F=(f_{y,y'})_{y,y'\in Y}\) by
\begin{equation*}
e_{x,x'}:=\sum_{y,y'}d_{x,y,x',y'},\quad f_{y,y'}:=\sum_{x,x'}d_{x,y,x',y'}.
\end{equation*}
so that \(E\in\mathrm{JM}(A,A)\) and \(F\in\mathrm{JM}(B,B)\).
The intersubjectivity of \(A\) and \(B\) gives
\begin{equation*}
\sum_{x\in X}e_{x,x}\ge\alpha1_S,\quad\sum_{y\in Y}f_{y,y}\ge\beta1_S.
\end{equation*}
Therefore,
\begin{equation*}
\begin{aligned}
\sum_{(x,y)\in X\times Y}d_{x,y,x,y}&=\sum_{x,y,y'}d_{x,y,x,y'}+\sum_{x,y,x'}d_{x,y,x',y}\\
&\qquad{}-\sum_{\substack{x,y,x',y'\\x=x'\text{ or }y=y'}}d_{x,y,x',y'}\\
&\ge(\alpha+\beta-1)1_S.
\end{aligned}
\end{equation*}
This completes the proof.
\item[(ii)]
We first prove the statement for intersubjectivity.
Suppose that \(C\) is \((1-\lambda(1-\alpha))\)-intersubjective.
Let \(D=(d_{x,x'})_{x,x'\in X}\in \mathrm{JM}(A,A)\) and \(E=(e_{x,x'})_{x,x'\in X}\in\mathrm{JM}(B,B)\).
Then we have \(\lambda D+(1-\lambda)E\in\mathrm{JM}(C,C)\).
Thus,
\begin{equation*}
\sum_{x\in X}\bigl(\lambda d_{x,x}+(1-\lambda)e_{x,x}\bigr)\ge (1-\lambda(1-\alpha))1_S.
\end{equation*}
Therefore,
\begin{equation*}
\sum_{x\in X}d_{x,x}
\ge\frac{1-\lambda(1-\alpha)}{\lambda}1_S-\frac{1-\lambda}\lambda\sum_{x\in X}e_{x,x}
\ge\alpha 1_S.
\end{equation*}
The statement for complete intersubjectivity follows by applying the same argument to each coarse-graining.
\qedhere
\end{itemize}
\end{proof}

Next, we present detailed calculations of the degree of (complete) intersubjectivity in several simple cases.

\begin{prop}
\label{intersubjectivity of coin-tossing measurements}
A coin-tossing measurement \(A=(\lambda_x1_S)_{x\in X}\) is (completely) \(\alpha\)-intersubjective if and only if
\begin{equation}
\alpha\le\max\Bigl\{2\max_{x\in X}\lambda_x-1,0\Bigr\}.
\end{equation}
\end{prop}

\begin{proof}
Without loss of generality, let \(X=\{1,\dots,n\}\) and assume that \(\lambda_1\ge\dots\ge\lambda_n\).
We distinguish two cases.

Case 1: \(\lambda_1\ge\frac12\).
For any \(B=(b_{i,j})_{i,j=1}^n\in\mathrm{JM}(A,A)\), we have
\begin{equation*}
\sum_{i=1}^nb_{i,i}\ge b_{1,1}=\lambda_11_S-\sum_{i=2}^nb_{1,i}\ge\lambda_11_S-\sum_{i=2}^n\lambda_i1_S=(2\lambda_1-1)1_S.
\end{equation*}
Conversely, define the measurement \(B=(b_{i,j})_{i,j=1}^n\) by
\begin{equation*}
b_{i,j}:=
\begin{cases}
(2\lambda_1-1)1_S&\text{if }i=j=1,\\
\lambda_j1_S&\text{if }i=1\text{ and }j\ge2,\\
\lambda_i1_S&\text{if }i\ge2\text{ and }j=1,\\
0_S&\text{otherwise}.
\end{cases}
\end{equation*}
Then, one verifies that \(B\in\mathrm{JM}(A,A)\) and \(\sum_{i=1}^nb_{i,i}=(2\lambda_1-1)1_S\).
Hence, \(A\) is \(\alpha\)-intersubjective if and only if \(\alpha\le2\lambda_1-1\).

Case 2: \(\lambda_1<\frac12\).
We show that there exists \(B=(b_{i,j})_{i,j=1}^n\in\mathrm{JM}(A,A)\) with \(\sum_{i=1}^nb_{i,i}=0_S\) by induction on \(n\).
For \(n=1,2\), the condition \(\lambda_1<\frac12\) cannot occur.
Assume \(n\ge3\).
Then there exists \(\mu>0\) such that \(\sum_{i=2}^n\min\{\lambda_i,\mu\}=1-2\lambda_1\).
Consider the coin-tossing measurement \(C:=\bigl(\frac{\min\{\lambda_i,\mu\}}{1-2\lambda_1}1_S\bigr)_{i=1}^n\).
By the definition of \(\mu\), its maximal probability does not exceed \(\frac12\); hence, by the induction hypothesis, there exists \(D=(d_{i,j})_{i,j=2}^n\in\mathrm{JM}(C,C)\) with \(\sum_{i=2}^nd_{i,i}=0_S\).
Define the measurement \(B=(b_{i,j})_{i,j=1}^n\) by
\begin{equation*}
b_{i,j}:=
\begin{cases}
0_S&\text{if }i=j=1,\\
(\lambda_j-\min\{\lambda_j,\mu\})1_S&\text{if }i=1\text{ and }j\ge2,\\
(\lambda_i-\min\{\lambda_i,\mu\})1_S&\text{if }i\ge2\text{ and }j=1,\\
(1-2\lambda_1)d_{i,j}&\text{otherwise}.
\end{cases}
\end{equation*}
Then, one verifies that \(B\in\mathrm{JM}(A,A)\) and \(\sum_{i=1}^nb_{i,i}=0_S\).
Hence, \(A\) is \(\alpha\)-intersubjective if and only if \(\alpha\le0\).

Combining the two cases, we conclude that \(A\) is \(\alpha\)-intersubjective if and only if
\(\alpha\le\max\{2\lambda_1-1,0\}\).
Since \(\max_{x\in X}\lambda_x\) is non-decreasing under coarse-graining, the same formula holds for complete intersubjectivity.
\end{proof}

\begin{prop}
\label{calculation in classical system}
In a classical system \(S\), a measurement \(A=(a_x)_{x\in X}\) is (completely) \(\alpha\)-intersubjective if and only if
\begin{equation}
\alpha\le\inf_{s\in\mathrm{ext}(S)}\max\Bigl\{2\max_{x\in X}a_x(s)-1,0\Bigr\}.
\end{equation}
\end{prop}

\begin{proof}
For each \(s\in\mathrm{ext}(S)\), \((a_x(s))_{x\in X}\) is a probability distribution.
Constructing a joint distribution \((b_{x,x'}(s))_{x,x'\in X}\) for it exactly as in the proof of Proposition \ref{intersubjectivity of coin-tossing measurements}, we obtain that the maximal value of \(\sum_{x\in X}b_{x,x}(s)\) equals \(\max\{2\max_{x\in X}a_x(s)-1,0\}\).
Performing this construction for each \(s\in\mathrm{ext}(S)\) and affinely extending it to all states in \(S\), we obtain a measurement
\(B=(b_{x,x'})_{x,x'\in X}\), which belongs to \(\mathrm{JM}(A,A)\).
Hence, \(A\) is \(\alpha\)-intersubjective if and only if
\(\alpha\le\inf_{s\in\mathrm{ext}(S)}\max\{2\max_{x\in X}a_x(s)-1,0\}\).
Since \(\inf_{s\in\mathrm{ext}(S)}\max_{x\in X}a_x(s)\) is non-decreasing under coarse-graining, the same formula holds for complete intersubjectivity.
\end{proof}

\begin{prop}
\label{intersubjectivity of unbiased measurement}
In a qubit system, an unbiased two-outcome measurement \(A=\bigl(\frac12(I+\bm\lambda\cdot\bm\sigma),\frac12(I-\bm\lambda\cdot\bm\sigma)\bigr)\) is (completely) \(\alpha\)-intersubjective if and only if
\begin{equation}
\alpha\le\abs{\bm\lambda}^2.
\end{equation}
\end{prop}

\begin{proof}
By Lemma 5 in the End Matter, it suffices to show that \(A\) is \(\alpha\)-sharp if and only if \(\alpha\le\frac{1+\abs{\bm\lambda}^2}2\), that is, that the optimal value of the minimization problem
\begin{equation*}
\begin{aligned}
\min_b\quad&1-\norm{b}_\infty\\
\text{subject to}\quad&0\le b,\quad b\le\frac12(I+\bm\lambda\cdot\bm\sigma),\quad b\le\frac12(I-\bm\lambda\cdot\bm\sigma)
\end{aligned}
\end{equation*}
is \(\frac{1+\abs{\bm\lambda}^2}2\).
Expanding \(b\) in the Pauli basis as \(b=b_0I+\bm b\cdot\bm\sigma\), we obtain the equivalent optimization problem
\begin{equation*}
\begin{aligned}
\min_{b_0,\bm b}\quad&1-b_0-\abs{\bm b}\\
\text{subject to}\quad&\abs{\bm b}\le b_0,\quad b_0\le\frac12-\abs{\bm b-\frac12\bm\lambda},\quad b_0\le\frac12-\abs{\bm b+\frac12\bm\lambda}.
\end{aligned}
\end{equation*}
At the optimum, \(b_0\) must saturate the tightest upper bound, namely, \(b_0=\min\bigl\{\frac12-\abs{\bm b-\frac12\bm\lambda},\frac12-\abs{\bm b+\frac12\bm\lambda}\bigr\}\).
By symmetry under the transformation \(\bm b\mapsto-\bm b\), it suffices to consider the case \(\bm\lambda\cdot\bm b\ge0\).
Then, the problem reduces to
\begin{equation*}
\begin{aligned}
\min_{\bm b}\quad&\frac12+\abs{\bm b+\frac12\bm\lambda}-\abs{\bm b}\\
\text{subject to}\quad&\abs{\bm b}+\abs{\bm b+\frac12\bm\lambda}\le\frac12,\quad\bm\lambda\cdot\bm b\ge0.
\end{aligned}
\end{equation*}
Therefore, as illustrated in Figure \ref{fig:proof of intersubjectivity of unbiased measurement}, in the shaded region bounded by the ellipse with foci at the origin and \(-\frac12\bm\lambda\) and the line perpendicular to its major axis, the optimal solution \(\bm b^*\) is the point that minimizes the difference between the distances to these two points.
At this point, the value of \(\frac12+\abs{\bm b^*+\frac12\bm\lambda}-\abs{\bm b^*}\) is equal to \(\frac{1+\abs{\bm\lambda}^2}2\).
\qedhere

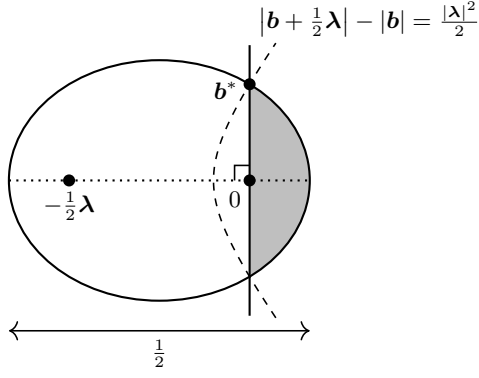
\begin{figure}
\begin{tikzpicture}[scale=0.8]
\path[use as bounding box](-4,-3.1)rectangle(2,2.9);
\begin{scope}
\clip(-1.5,0)circle[x radius=2.5,y radius=2];
\fill[gray,opacity=0.5](0,-2)rectangle(1,2);
\end{scope}
\draw[thick](-1.5,0)circle[x radius=2.5,y radius=2];
\draw[thick](0,-2.25)--(0,2.25);
\draw[thick,dotted](-4,0)--(1,0);
\draw[semithick](0,0.25)--(-0.25,0.25)--(-0.25,0);
\draw[semithick,dashed]plot[domain=-1.4:1.4,smooth]({0.9*cosh(\x)-1.5},{1.2*sinh(\x)});
\draw[above right](0,2.25)node{\(\abs{\bm b+\tfrac12\bm\lambda}-\abs{\bm b}=\tfrac{\abs{\bm\lambda}^2}2\)};
\fill(0,0)circle[radius=0.1];
\draw[below left](0,-0.05)node{\(0\)};
\fill(-3,0)circle[radius=0.1];
\draw[below](-3,-0.05)node{\(-\frac12\bm\lambda\)};
\fill(0,1.6)circle[radius=0.1];
\draw[left](-0.05,1.5)node{\(\bm b^*\)};
\draw[semithick,{<[scale=1.2]}-{>[scale=1.2]}](-4,-2.5)--(1,-2.5);
\draw[below](-1.5,-2.5)node{\(\frac12\)};
\end{tikzpicture}
\caption{
Geometric illustration of the optimization problem in Proposition \ref{intersubjectivity of unbiased measurement}.
}
\label{fig:proof of intersubjectivity of unbiased measurement}
\end{figure}

\end{proof}

\section{Proofs of Example 3 and Example 4}

In this appendix, we provide proofs of Example 3 and Example 4 presented in the End Matter.

\begin{proof}[Proof of Example 3]
Let
\begin{equation*}
S:=\left\{(x_1,x_2,x_3)\in\mathbb R^3\:\middle|\:-x_3\le x_1\le x_3,\ -x_3\le x_2\le x_3,\ x_3=1\right\},
\end{equation*}
and let \(A:=\frac12B+\frac12C\), where \(B=(b_+,b_-)\) and \(C=(c_+,c_-)\) are defined by
\begin{equation*}
b_\pm(x_1,x_2,x_3):=\frac12(x_3\pm x_1),\quad c_\pm(x_1,x_2,x_3):=\frac12(x_3\pm x_2).
\end{equation*}
We show that \(A\) is completely intersubjective but not extremal.
It is clear that \(A\) is not extremal.
Hence, it suffices to show that \(A\) is sharp.
Suppose that an effect \(d\) satisfies
\begin{equation*}
d\le\frac12b_++\frac12c_+,\quad d\le\frac12b_-+\frac12c_-.
\end{equation*}
Evaluating these inequalities at the states \((-1,-1,1)\) and \((1,1,1)\), respectively, gives \(d(-1,-1,1)=d(1,1,1)=0\).
Moreover, since \(d\) is affine, we have \(d(1,-1,1)+d(-1,1,1)=d(-1,-1,1)+d(1,1,1)=0\).
By positivity of \(d\), this implies \(d(1,-1,1)=d(-1,1,1)=0\).
Therefore, \(d=0_S\), and hence \(A\) is sharp.
\end{proof}

\begin{proof}[Proof of Example 4]
Let
\begin{equation*}
S:=\left\{(x_1,x_2,x_3,x_4,x_5)\in\mathbb R^5\:\middle|\:0\le x_i\ (i=1,\dots,5),\ x_5\le x_1+x_4,\ x_5\le x_2+x_3,\ x_1+x_2+x_3+x_4=1\right\}.
\end{equation*}
Let \(A=(a_1,a_2,a_3,a_4)\), where
\begin{equation*}
a_i(x_1,x_2,x_3,x_4,x_5):=x_i\quad(i=1,2,3,4).
\end{equation*}
We also define \(B:=(a_1+a_2,a_3+a_4)\), \(C:=(a_1+a_3,a_2+a_4)\), and \(s_i:=(\delta_{1,i},\delta_{2,i},\delta_{3,i},\delta_{4,i},0)\in S\) for \(i=1,2,3,4\).

We first show that \(A\) is elementwise sharp and extremal, but not completely intersubjective.
To prove elementwise sharpness, suppose that an effect \(e\) satisfies \(e\le a_1\) and \(e\le1_S-a_1\).
Then \(e\) must be of the form
\begin{equation*}
e(x_1,x_2,x_3,x_4,x_5)=\lambda x_5
\end{equation*}
for some constant \(\lambda\).
Positivity of \(e\) gives \(\lambda\ge0\), while evaluating \(e\le a_1\) at \(\bigl(0,0,\frac12,\frac12,\frac12\bigr)\in S\) gives \(\lambda\le0\).
Thus \(\lambda=0\), and hence \(e=0_S\).
Therefore \(a_1\) is a sharp effect.
The same argument applies to \(a_2\), \(a_3\), and \(a_4\).

The extremality of \(A\) follows from the fact that \(a_1,a_2,a_3,a_4\) are linearly independent indecomposable effects.

On the other hand, \(A\) is not completely intersubjective. Indeed, the effect \(e\), defined by
\begin{equation*}
e(x_1,x_2,x_3,x_4,x_5):=x_5,
\end{equation*}
satisfies \(e\le a_1+a_4\) and \(e\le a_2+a_3\).
Thus the coarse-graining \(D:=(a_1+a_4,a_2+a_3)\) of \(A\) is not sharp, and hence \(A\) is not completely intersubjective.

We next show that \(B\) and \(C\) are completely intersubjective.
Suppose that an effect \(e\) satisfies \(e\le a_1+a_2\) and \(e\le a_3+a_4\).
Then \(e\) must be of the form
\begin{equation*}
e(x_1,x_2,x_3,x_4,x_5)=\lambda x_5
\end{equation*}
for some constant \(\lambda\).
Positivity of \(e\) gives \(\lambda\ge0\), while evaluating \(e\le a_1+a_2\) at \(\bigl(0,0,\frac12,\frac12,\frac12\bigr)\in S\) gives \(\lambda\le0\).
Thus \(\lambda=0\), and hence \(e=0_S\).
Therefore, \(B\) is sharp, and hence completely intersubjective.
The same argument applies to \(C\).

Finally, we show that \((s_1,s_2,s_3,s_4)\) is perfectly distinguished only by \(A\).
Let \(A'=(a'_1,a'_2,a'_3,a'_4)\) be a measurement that perfectly distinguishes these states. Then each \(a'_i\) must be of the form
\begin{equation*}
a'_i(x_1,x_2,x_3,x_4,x_5)=x_i+\lambda_ix_5
\end{equation*}
for some constant \(\lambda_i\).
Positivity of \(a'_i\) implies \(\lambda_i\ge0\) for each \(i\).
Since \(A'\) is a measurement, we have \(\lambda_1+\lambda_2+\lambda_3+\lambda_4=0\).
Therefore \(\lambda_1=\lambda_2=\lambda_3=\lambda_4=0\), which implies \(A'=A\).
\end{proof}

\section{Results for continuous-outcome measurements}
\label{Results for continuous-outcome measurements}

Several of our results presented in the main text remain valid for continuous-outcome measurements.
In this appendix, we extend our analysis to this more general setting.

A continuous-outcome measurement is described by an Effect-Valued Measure (EVM) defined on a measurable space \((X,\Sigma_X)\).
Here, \(X\) denotes the outcome space and \(\Sigma_X\) the \(\sigma\)-algebra of its measurable subsets.
An EVM on \(X\) is a map \(A:\Sigma_X\to\mathcal E_S\) such that \cite{Davies1970}:
\begin{enumerate}
\item \(A(X)=1_S\);
\item For any countable collection of mutually disjoint sets \(\{U_n\}_{n\in\mathbb{N}}\subseteq\Sigma_X\),
\begin{equation}
A\Biggl(\bigcup_{n\in\mathbb{N}}U_n\Biggr)=\sum_{n\in\mathbb{N}}A(U_n),
\end{equation}
\textit{i.e.}, \(A\bigl(\bigcup_{n\in\mathbb{N}}U_n\bigr)(s)=\sum_{n\in\mathbb{N}}A(U_n)(s)\) for every state \(s\in S\).
\end{enumerate}

For a state \(s\in S\), the map \(U\mapsto A(U)(s)\) defines a probability measure on \(X\), representing the outcome statistics of the measurement.
Finite-outcome measurements, which we discussed in the main text, correspond to the special case in which \(X\) is a finite set equipped with the discrete \(\sigma\)-algebra.
As in the finite-outcome case, we denote by \(\mathcal M_S(X)\) the set of all EVMs on \(X\).

A measurement \(C\in\mathcal M_S(X\times Y)\) is called a \textit{joint measurement} of a pair of measurements \(A\in\mathcal M_S(X)\) and \(B\in\mathcal M_S(Y)\) if its marginals reproduce \(A\) and \(B\), that is,
\begin{equation}
\forall U\in\Sigma_X,\quad C(U\times Y)=A(U),\quad\text{and}\quad\forall V\in\Sigma_Y,\quad C(X\times V)=B(V).
\end{equation}
We denote the set of all joint measurements of \(A\) and \(B\) by \(\mathrm{JM}(A,B)\).

A measurement \(B\in\mathcal M_S(Y)\) is called a \textit{coarse-graining} of a measurement \(A\in\mathcal M_S(X)\) if there exists a measurable map \(\pi:X\to Y\) such that
\begin{equation}
\forall V\in\Sigma_Y,\quad A\bigl(\pi^{-1}(V)\bigr)=B(V).
\end{equation}

We now introduce the notion of (complete) intersubjectivity for continuous-outcome measurements.

\begin{dfn}
A measurement \(A\in\mathcal M_S(X)\) is said to be \textit{intersubjective} if \(\mathrm{JM}(A,A)\) is a singleton, that is, if every \(B\in\mathrm{JM}(A,A)\) coincides with the canonical joint measurement:
\begin{equation}
\forall U_1,U_2\in\Sigma_X,\quad B(U_1\times U_2)=A(U_1\cap U_2).
\end{equation}
We say that \(A\) is \textit{completely intersubjective} if every coarse-graining of \(A\) is intersubjective.
\end{dfn}

\begin{lem}
A measurement \(A\in\mathcal M_S(X)\) is completely intersubjective if and only if \(A(U)\) is a sharp effect for every \(U\in\Sigma_X\).
\end{lem}

\begin{proof}
If \(A\) is completely intersubjective, then for any \(U\in\Sigma_X\), the two-outcome coarse-graining \(B:=(A(U),A(X\setminus U))\) is intersubjective.
Therefore, \(A(U)\) is a sharp effect.

Conversely, suppose that \(A(U)\) is a sharp effect for every \(U\in\Sigma_X\).
Let \(B\in\mathcal M_S(Y)\) be a coarse-graining of \(A\) via a measurable map \(\pi:X\to Y\), and let \(C\in\mathrm{JM}(B,B)\).
For any \(V_1,V_2\in\Sigma_Y\), we have
\begin{equation*}
C(V_1\times V_2)\le B(V_1)=A\bigl(\pi^{-1}(V_1)\bigr),\quad C(V_1\times V_2)\le B(V_2)=A\bigl(\pi^{-1}(V_2)\bigr).
\end{equation*}
In particular, if \(V_1\cap V_2=\emptyset\), then \(\pi^{-1}(V_1)\cap \pi^{-1}(V_2)=\emptyset\), and the sharpness of \(A\bigl(\pi^{-1}(V_1)\bigr)\) implies \(C(V_1\times V_2)=0_S\).
For arbitrary \(V_1,V_2\in\Sigma_Y\), it then follows that
\begin{equation*}
\begin{aligned}
C(V_1\times V_2)&=C((V_1\cap V_2)\times V_2)+C((V_1\setminus V_2)\times V_2)\\
&=C((V_1\cap V_2)\times V_2)\\
&\le B(V_1\cap V_2).
\end{aligned}
\end{equation*}
Summing this inequality over all choices of complements of \(V_1\) and \(V_2\), both sides add up to \(1_S\).
Hence, equality must hold term by term:
\begin{equation*}
\forall V_1,V_2\in\Sigma_Y,\quad C(V_1\times V_2)=B(V_1\cap V_2).
\end{equation*}
Since the sets of the form \(V_1\times V_2\) generate the \(\sigma\)-algebra on \(Y\times Y\), Carath\'eodory's extension theorem ensures that \(C\) coincides with the canonical joint measurement of \(B\) with itself.
Therefore, \(B\) is intersubjective.
\end{proof}

From the fact that in quantum theory, an effect is sharp if and only if it is represented by a projection, we obtain the following characterization of PVMs.

\begin{thm}[Continuous-outcome version of Theorem 3]
A (continuous-outcome) POVM is a PVM if and only if it is completely intersubjective.
\end{thm}

Next, we summarize some basic properties of (completely) intersubjective measurements.

\begin{prop}
Suppose that two measurements \(A\) and \(B\) are intersubjective.
Then any joint measurement \(C\in\mathrm{JM}(A,B)\) is also intersubjective.
\end{prop}

\begin{proof}
Suppose that \(A\in\mathcal M_S(X)\) and \(B\in\mathcal M_S(Y)\).
Let \(D\in\mathrm{JM}(C,C)\).
Define two measurements \(E\in\mathcal M_S(X\times X)\) and \(F\in\mathcal M_S(Y\times Y)\) by
\begin{equation*}
E(U_1\times U_2):=D(U_1\times Y\times U_2\times Y),\quad F(V_1\times V_2):=D(X\times V_1\times X\times V_2),
\end{equation*}
so that \(E\in\mathrm{JM}(A,A)\) and \(F\in\mathrm{JM}(B,B)\).
The intersubjectivity of \(A\) and \(B\) gives
\begin{equation*}
\forall U_1,U_2\in\Sigma_X,\quad E(U_1\times U_2)=A(U_1\cap U_2),\quad\text{and}\quad\forall V_1,V_2\in\Sigma_Y,\quad F(V_1\times V_2)=B(V_1\cap V_2).
\end{equation*}
It then follows that
\begin{equation*}
\begin{aligned}
D(U_1\times V_1\times U_2\times V_2)&\le D((U_1\cap U_2)\times(V_1\cap V_2)\times U_2\times V_2)+E((U_1\setminus U_2)\times U_2)+F((V_1\setminus V_2)\times V_2)\\
&=D((U_1\cap U_2)\times(V_1\cap V_2)\times U_2\times V_2)\\
&\le C((U_1\cap U_2)\times(V_1\cap V_2))\\
&=C((U_1\times V_1)\cap(U_2\times V_2)).
\end{aligned}
\end{equation*}
Summing this inequality over all choices of complements of \(U_1\), \(V_1\), \(U_2\), and \(V_2\), both sides add up to \(1_S\).
Hence, equality must hold term by term:
\begin{equation*}
\forall U_1,U_2\in\Sigma_X,\ \forall V_1,V_2\in\Sigma_Y,\quad D(U_1\times V_1\times U_2\times V_2)=C((U_1\times V_1)\cap(U_2\times V_2)).
\end{equation*}
Since the sets of the form \(U_1\times V_1\times U_2\times V_2\) generate the \(\sigma\)-algebra on \(X\times Y\times X\times Y\), Carath\'eodory's extension theorem ensures that \(D\) coincides with the canonical joint measurement of \(C\) with itself.
Therefore, \(C\) is intersubjective.
\end{proof}

\begin{prop}
Let \(C=\lambda A+(1-\lambda)B\) with \(\lambda\in (0,1]\) be a probabilistic mixture of two measurements \(A\) and \(B\).
If \(C\) is (completely) intersubjective, then \(A\) is also (completely) intersubjective.
\end{prop}

\begin{proof}
We first prove the statement for intersubjectivity.
Let \(D_1,D_2\in\mathrm{JM}(A,A)\) and \(E\in\mathrm{JM}(B,B)\).
Then \(\lambda D_1+(1-\lambda)E\) and \(\lambda D_2+(1-\lambda)E\) both belong to \(\mathrm{JM}(C,C)\).
Since \(C\) is intersubjective, we must have \(\lambda D_1+(1-\lambda)E=\lambda D_2+(1-\lambda)E\), which implies \(D_1=D_2\).
Hence, \(A\) is intersubjective.

The statement for complete intersubjectivity follows by applying the same argument to each coarse-graining.
\end{proof}

To connect the intersubjectivity of measurements with the extremality, we first establish a useful characterization of extremal measurements in terms of the uniqueness of joint measurements.

\begin{lem}
\label{extremal measurements}
For a measurement \(A\in\mathcal M_S(X)\), the following three conditions are equivalent:
\begin{itemize}
\item[(i)] \(A\) is an extreme point of \(\mathcal M_S(X)\).
\item[(ii)] For every measurement \(B\), if a joint measurement of \(A\) and \(B\) exists, it is unique.
\item[(iii)] There exists a nontrivial coin-tossing measurement \(B\) such that the joint measurement of \(A\) and \(B\) is unique.
Here, a measurement \(B\in\mathcal M_S(Y)\) is called a coin-tossing measurement if there exists a probability measure \(\mu\) on \(Y\) such that \(B(V)=\mu(V)1_S\) for all \(V\in\Sigma_Y\).
It is called nontrivial if there exists \(V\in\Sigma_Y\) with \(0<\mu(V)<1\).
\end{itemize}
\end{lem}

\begin{proof}
(i)\(\implies\)(ii).
Let \(B\in\mathcal M_S(Y)\), and suppose \(C_1,C_2\in\mathrm{JM}(A,B)\).
Fix \(V_0\in\Sigma_Y\) and define two measurements \(D_1,D_2\in\mathcal M_S(X)\) by
\begin{equation*}
\begin{aligned}
D_1(U)&:=C_1(U\times V_0)+C_2(U\times(Y\setminus V_0))=A(U)+C_1(U\times V_0)-C_2(U\times V_0),\\
D_2(U)&:=C_1(U\times(Y\setminus V_0))+C_2(U\times V_0)=A(U)-C_1(U\times V_0)+C_2(U\times V_0).
\end{aligned}
\end{equation*}
Then \(A=\frac12D_1+\frac12D_2\).
Since \(A\) is extremal in \(\mathcal M_S(X)\), we must have \(D_1=D_2\), implying that \(C_1(U\times V_0)=C_2(U\times V_0)\) for all \(U\in\Sigma_X\).
As this holds for arbitrary \(V_0\in\Sigma_Y\), we conclude \(C_1=C_2\).
Hence, the joint measurement is unique.

(ii)\(\implies\)(iii).
Immediate.

(iii)\(\implies\)(i).
Suppose that there exists a nontrivial coin-tossing measurement \(B\in\mathcal M_S(Y)\) with corresponding probability measure \(\mu\) (\textit{i.e.} \(B(V)=\mu(V)1_S\)) such that the joint measurement of \(A\) and \(B\) is unique.
Since \(B\) is nontrivial, there exists \(V_0\in\Sigma_Y\) such that \(0<\mu(V_0)<1\).
Suppose that \(A=\mu(V_0)C_1+(1-\mu(V_0))C_2\) for some \(C_1,C_2\in\mathcal M_S(X)\).
Define \(D\in\mathcal M_S(X\times Y)\) by
\begin{equation*}
D(U\times V):=\mu(V\cap V_0)C_1(U)+\mu(V\setminus V_0)C_2(U).
\end{equation*}
Then \(D\) is a joint measurement of \(A\) and \(B\).
By the uniqueness assumption, this joint measurement must coincide with the canonical one: \(D(U\times V)=A(U)B(V)=\mu(V)A(U)\).
In particular, evaluating \(D\) on \(U\times V_0\), we conclude that \(C_1(U)=A(U)\) for all \(U\in\Sigma_X\).
Hence, \(A\) is an extreme point of \(\mathcal M_S(X)\).
\end{proof}

\begin{rmk}
Lemma \ref{extremal measurements} extends the result of Ref.~\cite{Guerini2018} to continuous-outcome measurements.
The same argument also applies when \(A\) is taken to be a general channel rather than a measurement.
Moreover, condition (ii) remains valid when \(B\) is replaced by an arbitrary channel.
In this setting, Lemma \ref{extremal measurements} further shows that if \(B\) is an extremal channel, then the joint channel, when it exists, is extremal as well.
These results provide a characterization of extremal channels and generalize the results of Ref.~\cite{Haapasalo2014} to the framework of GPTs.
\end{rmk}

Applying condition (ii) of Lemma \ref{extremal measurements} to the case \(B:=A\) yields the following proposition:

\begin{prop}
Every extreme point of \(\mathcal M_S(X)\) is an intersubjective measurement.
\end{prop}

\end{document}